\documentclass{article}
\usepackage{graphicx} 
\usepackage{amsmath,amssymb,amsthm}
\usepackage{authblk}
\usepackage[ruled,noend,linesnumbered]{algorithm2e} 
\usepackage{cleveref}
\usepackage{complexity}
\usepackage{enumitem} 

\usepackage{todonotes}
\presetkeys{todonotes}{inline}{}

\usepackage[round]{natbib}

\newtheorem{theorem}{Theorem}[section]
\newtheorem{corollary}[theorem]{Corollary}
\newtheorem{proposition}[theorem]{Proposition}

\newtheorem{remark}[theorem]{Remark}

\theoremstyle{definition}

\newtheorem{example}[theorem]{Example}

\usepackage{fullpage}

\newcommand{\allones}{\mathbf{1}}
\newcommand{\allzeros}{\mathbf{0}}
\newcommand{\unitvect}{\mathbf{e_1}}
\newcommand{\supt}{^\intercal}
 \def\1{\1}

\newcommand{\numbpara}{\lambda}

\newcommand{\Col}{\mbox{\textsc{Col}}}
\newcommand{\Row}{\mbox{\textsc{Row}}}

\newcommand{\surplus}{\phi}

\title{A Generalization of von Neumann's Reduction from
the Assignment Problem to Zero-Sum Games}
\date{}

\author[1]{Ilan Adler}
\author[2]{Martin Bullinger}
\author[3]{Vijay V. Vazirani}

\affil[1]{ \small Department of Industrial Engineering and Operations Research, University of California Berkeley, USA}
\affil[2]{ \small School of Engineering Mathematics and Technology, University of Bristol, UK}
\affil[3]{ \small Department of Computer Science, University of California Irvine, USA\protect\\ \vspace*{0.1cm} adler@ieor.berkeley.edu, martin.bullinger@bristol.ac.uk, vazirani@ics.uci.edu}

\begin{document}

\maketitle

\begin{abstract}
The equivalence between von Neumann's Minimax Theorem for zero-sum games and the LP Duality Theorem connects cornerstone problems of the two fields of game theory and optimization, respectively, and has been the subject of intense scrutiny for seven decades. Yet, as observed in this paper, the proof of the difficult direction of this equivalence is unsatisfactory: 
It does not assign distinct roles to the two players of the game, as is natural from the definition of a zero-sum game. 

In retrospect, a partial resolution to this predicament was provided in another brilliant paper of \citet{vNeu53a}, which reduced the assignment problem to zero-sum games. However, the underlying LP is highly specialized --- all entries of its objective function vector are strictly positive, the constraint vector is all ones, and the constraint matrix is 0/1.

We generalize von Neumann's result along two directions, each allowing negative entries in certain parts of the LP. Our reductions make explicit the roles of the two players of the reduced game, namely their maximin strategies are to play optimal solutions to the primal and dual LPs. 
Furthermore, unlike previous reductions, the value of the reduced game reveals the value of the given LP. 
Our generalizations encompass several basic economic scenarios. 
\end{abstract}

\section{Introduction}
\label{sec:intro}

The Minimax Theorem for zero-sum games \citep{vNeu28a}, which is regarded as one of \citeauthor{vNeu28a}'s spectacular contributions, not only started off\footnote{Von Neumann is famously quoted as saying ``As far as I can see, there could be no theory of games [\dots] without that theorem [\dots]. I thought there was nothing worth publishing until the Minimax Theorem was proved'' \citep[p.~19]{Casti-Minimax}.} game theory \citep{von1947theory}  
but also had far-reaching consequences in other areas.
In particular, it is equivalent to the Linear Programming Duality Theorem and it led to Yao's Lemma \citep{Yao1977probabilistic}, a key tool for proving a lower bound on the worst-case performance of a randomized algorithm. The equivalence mentioned above is particularly significant since it connects cornerstone problems of the two fields of game theory and optimization. 
Our paper provides a fresh insight into this equivalence.

The two directions of this equivalence are qualitatively different: It is straightforward to reduce a zero-sum game to a primal-dual pair of linear programs (LPs) whose optimal solutions yield maximin strategies of the game \citep{Dant51a}. 
However, the other direction is much harder, i.e., the task of 
solving a linear program --- or deciding that the primal or the dual program is infeasible --- by finding maximin strategies of a corresponding zero-sum game. 
This direction follows from a series of papers published over a course of more than seven decades, starting with the seminal paper by \citet{Dant51a}.\footnote{Interestingly, \citet{Dant82a} describes a memorable meeting of his with von Neumann in which the latter outlined such a correspondence. \citet{Dant82a} also attributes the LP Duality Theorem to von Neumann, again as a consequence of the same meeting, even though the first written proof is due to Tucker and his students Kuhn and Gale.}  However, the current reduction has an unsatisfactory aspect: it does not assign distinct roles to the two players of the reduced game, as is natural from the definition of a zero-sum game, see details below. 

In retrospect, a partial resolution to this quandary was provided in another brilliant paper of \citet{vNeu53a}, which reduced the assignment problem to zero-sum games; he called the latter the {\em hide-and-seek game}. Its drawback is that the underlying LP is highly specialized, see below. In this paper, we  generalize  von Neumann's reduction along two directions, to two fairly large classes of LPs. Our generalizations capture several basic economic scenarios, beyond the assignment problem, and  our reduction has the following properties:
\begin{enumerate}
    \item It uses a scaled version of the constraint matrix as the game matrix, see details in the discussion following Theorem \ref{thm:main}.
    \item The row and column players obtain their maximin strategies by separately finding optimal solutions to the primal and dual LPs, respectively, thereby making explicit the roles of the two players. 
    \item If the LP has an optimal solution then the value of the reduced game is the reciprocal of the value of the LP. 
\end{enumerate}

We now examine in detail the current proofs of the two directions of the equivalence. For the easy direction, interpret the game matrix as the constraint matrix of a maximization LP and solve it optimally to obtain a maximin strategy for the row player \citep[see, e.g.,][Section~2]{vSte23a}. Next optimally solve the dual LP to obtain a maximin strategy for the column player. 
By strong duality, the two LPs have the same objective function value, hence showing that the payoffs of the two players from these strategies must be equal. 
Observe that in general, the two players have distinct roles: their maximin strategies are solutions to different LPs and are vectors of different dimensions.  

For the difficult direction, \citet{Dant51a} reduced the given primal-dual pair of LPs to a symmetric zero-sum game; such a game has symmetric maximin strategies and the value of the game is zero. 
In this game, each player simultaneously mimics both primal and dual LPs, i.e., a strategy of the row player (or the column player) can be separated into two parts, corresponding to solutions of the primal and dual LPs. Even though this result stood for over half a century, it is in fact incomplete as observed by \citet{Adle13a}: the constructed zero-sum game can have solutions of a particular form which cannot be interpreted as optimal primal-dual pairs or certificates of infeasibility of the primal or the dual LP. Later --- by adding an additional row to the zero-sum game --- this reduction was extended to work in full generality, first by \citet{Adle13a} and later, adding a certificate of infeasibility, by \citet{vSte23a} and \citet{BrRe23a}.

Next, we describe the partial resolution by \citet{vNeu53a}. Let $(G, \surplus)$ be an instance of the assignment problem on a complete bipartite graph $G$ with $n$ vertices on each side and a strictly positive weight $\surplus_{ij}$ on edge $(i, j)$. The transformed game matrix $M$ is $2n \times n^2$, where the $2n$ rows correspond to the $2n$ vertices of $G$ and the $n^2$ columns correspond to the $n^2$ edges of $G$. 
Let $\ell$ represent edge $(i, j)$ of $G$. 
Then the $(k, \ell)^{th}$ entry of $M$ is defined as follows:
        \[
            M(k, \ell) \ \ = \ \begin{cases}
                1/\surplus_{ij} & \text{if } k = i \ \text{ or } k = j, \\
                0 & \text{otherwise}.
            \end{cases}
        \]
Observe that each column has exactly two nonzero entries. 

The assignment problem asks for a maximum weight matching in bipartite graph $G$ under edge weights $\surplus_{ij}$. 
In the hide-and-seek game,
the row player picks one vertex and the column player picks one edge, say $e = (i,j)$, of~$G$.
If the vertex is one of the endpoints of the edge, then the column player pays $\frac 1{\surplus_{ij}}$ to the row player, and  otherwise, no payment is made. 
The row player attempts to maximize his payoff and the column player minimizes the amount she has to pay. 
The reason for the name ``hide-and-seek'' should now be clear: the column player hides in an edge and the row player attempts to guess one of its endpoints. 
As usual, mixed strategies will be probability distributions over pure strategies. 

Consider the primal-dual pair of LPs for the assignment problem in $G$: the primal LP finds a maximum weight fractional perfect matching in $G$ and the dual LP finds a minimum generalized vertex cover. The most remarkable step in the construction given above is von Neumann's choice of the payoff matrix. 
Indeed, up to a change in variables caused by scaling, this matrix is exactly the constraint matrix of the LP.
He shows that for this payoff matrix, the optimal solutions to the primal and dual LPs yield maximin strategies for the column and row players, respectively.
But why is the payment the reciprocal of the edge weight? 
Our reduction generalizes the idea of using the scaled constraint matrix as the game matrix.
From this larger perspective, we are able to dispel the mystery behind von Neumann's reduction, see \Cref{sec:productionapplication}.

A drawback of von Neumann's reduction is that his LPs are highly specialized and have almost a  ``combinatorial'' feel to them: His primal LP is $\max c\supt x \ \ \text{s.t. }  A x \le \allones, \ x \geq \allzeros$, with all entries in $c$ being strictly positive and the constraint matrix, $A$, being the incidence matrix of the underlying complete bipartite graph, i.e., all components of this matrix are zeros or ones. 
That still leaves a nagging doubt: can a separation of the roles of the two players be achieved if the LPs are more general, e.g., 
for a more general constraint vector $b$ or for data with negative entries. 
In this paper, we generalize von Neumann's hide-and-seek game along two directions: 

\begin{enumerate}
	\item  If $A$ is arbitrary, but $b$ and $c$ are strictly positive. 
	\item  If $b$ and $c$ are arbitrary, but $A$ is nonnegative. 
\end{enumerate}

A wide range of applications are covered by these two extensions. 
The allocation of scarce resources, including classic economic scenarios such as the diet problem (via its dual LP) or farm planning, can be modelled with LPs which typically have all components of $A$, $b$, and $c$ positive \citep[][Chapter~11]{Chva83a}, as required by the first extension.
Additional applications are covered by the second extension, e.g., the transportation problem which is equivalent to $b$-matching, a direct generalization of the assignment problem \citep[][Chapter~21]{Schr03a}.
In fact, the setting of LPs with positive $b$ and~$c$ corresponds to production planning scenarios and gives rise to a natural game-theoretic interpretation extending \citeauthor{vNeu53a}'s hide-and-seek game, as we discuss in \Cref{sec:productionapplication}.
We call the resulting game the \emph{production planning game}.

We leave the open problem of giving a complete generalization. 
We believe this will require a substantially new approach, since we seem to have fully extended von Neumann's idea of rescaling. 
Our results include: 

\begin{enumerate}
    \item We cover larger classes of linear programs, beyond the assignment problem, that capture many important applications.
    \item Our generalization even covers linear programs whose primal is unbounded (and hence dual is infeasible). 
    \item Our reduction carves out the exact reason why the hide-and-seek game uses the reciprocals of the assignment values for the payments.
\end{enumerate}

Finally, since the reduced game always has a well-defined value, achieved by maximin strategies, it offers a new perspective on LP pairs, one of which is unbounded. 
The game-theoretic world allows us to distinguish such LPs; this can be interpreted as a degree of (in)feasibility, see \Cref{sec:posvectors}.

\paragraph{Related Work}
Some previous work has proposed other generalizations of \citeauthor{vNeu53a}'s hide-and-seek game.
However, in contrast to our paper, none of these provides a reduction between linear programs and zero-sum games.

First, \citet{MaTa81a} generalized the hide-and-seek game to higher dimensions.
Recall that the original interpretation of \citeauthor{vNeu53a}'s hide-and-seek game is that the first player hides in one cell of a matrix, and the other player tries to find the first player by checking a row or a column of the matrix.
\citet{MaTa81a} study a game in which the first player hides in a higher-dimensional tensor.

Second, \citet{GaJa24a} study a different generalization of \citeauthor{vNeu53a}'s result.
In contrast to \citet{MaTa81a}, their reduction is from a version of the assignment problem with a more general, \emph{nonlinear} payoff structure.
They reduce to a two-player game generalizing the hide-and-seek game that is not zero-sum anymore.

\section{Preliminaries}
\label{sec:prelims}

We start with some conventions.
Whenever we have an inequality between vectors, we assume that it holds componentwise.
Moreover, for a vector $v$ and a matrix $M$, we denote by $v\supt$ and $M\supt$ its transposed vector and matrix, respectively.
Finally, we denote by~$\allones$ and~$\allzeros$ the column vectors or matrices consisting of $1$- and $0$-entries only; their dimension will be clear from the context.

We now introduce linear programs and zero-sum games, the main objects of our study.
For text book introductions to these topics, we recommend the books by \citet{Schr98a} and \citet{NRTV07a}.

We consider linear programs of the following form.
Assume we are given $A\in \mathbb R^{m\times n}$, $b\in \mathbb R^m$, and $c\in \mathbb R^n$. 
We denote their components by $a_{ij}$, $b_i$, and $c_j$, respectively.
These objects specify the LP $(P)$  together with its dual $(D)$.
\begin{align*}
    (P)&
    & (D) & \\
    \text{max} \quad c\supt x &
    &\text{min} \quad b\supt y & \\
    \text{s.t. } \quad A x &\le b
    &\text{s.t. } \quad A\supt y&\ge c \\
    x&\ge \allzeros
    & y &\geq \allzeros\\
\end{align*}

We write $(P,D)$ to refer to the pair of LPs $(P)$ and $(D)$.
A vector $x\in \mathbb R^n$ is called \emph{(primal) feasible} if it satisfies $Ax\le b$ and $x\ge \allzeros$, i.e., if it satisfies the constraints of $(P)$.
Similarly, a vector $y\in \mathbb R^m$ is called \emph{(dual) feasible} if it satisfies $A\supt y\ge c$ and $y\ge \allzeros$, i.e., if it satisfies the constraints of $(D)$.
The \emph{weak duality} theorem states that if $x$ and $y$ are primal feasible and dual feasible, respectively, then $c\supt x \le b\supt y$.
This implies that whenever we find feasible $x$ and $y$ such that $c\supt x = b\supt y$, then $x$ is a feasible vector maximizing $c\supt x$ and $y$ is a feasible vector minimizing $b\supt y$.
In this case, we say that $(x,y)$ an \emph{optimal primal-dual pair}.
The \emph{strong duality} theorem states that whenever $(P)$ and $(D)$ admit feasible points, then they admit an optimal primal-dual pair, and if so, $c\supt x = b\supt y$; this is called the {\em value of the LP}. 
The following characterization of possible outcomes of a pair of primal and dual LPs $(P,D)$ is well-known.

\begin{proposition}
\label{prop:LP_result}
    For any 
    pair of linear programs $(P,D)$, exactly one of the following outcomes is true:
    \begin{enumerate}
    \item \label{it:opt}
       There exist $x^*,y^*$  such that  
      $$Ax^* \leq b, \;x^*\geq \allzeros;\; A\supt y^* \geq c, \;y^*\geq \allzeros,\; c\supt x^*=b\supt y^*.$$
      In this case $x^*$ and $y^*$ are  optimal solutions to $(P)$ and $(D)$ respectively. 
       
      \item \label{it:inf}
         There exists $z^*$  such that   
      $$A\supt z^* \geq \allzeros,\;z^*\geq \allzeros,\;  b\supt z^*<0.$$
      In this case  $(P)$ is infeasible. We call $z^*$ an infeasibility certificate for $(P)$.   
    \item \label{it:unb}
    There exists $w^*$ such that 
         $$ A w^* \leq \allzeros,\; w^*\geq \allzeros,\; c\supt w^*>0.$$
      In this case  $(P)$ is unbounded if feasible. We call $w^*$ an unboundedness certificate for $(P)$.   
    \end{enumerate}
\end{proposition}

Given $A\in \mathbb R^{m\times n}$, $c\in \mathbb R^n$, and $b\in \mathbb R^m$, the {\it problem of solving the linear program $(P)$} is to provide one of the three outputs stated above.

\begin{remark}\label{rem:LPremarks}
We make a few important remarks.

$(i)$ In order to certify $(P)$ to be unbounded, we need to show, in addition to the unboundedness certificate, a feasible solution to $(P)$ (i.e., a vector~$x$ such that $Ax\le b$, $x\ge \allzeros)$. 

{$(ii)$} If $z^*$ is an infeasibility certificate for $(P)$, then $z^*$ is an unboundedness certificate for $(D)$.
Similarly, if $w^*$ is an unboundedness certificate for $(P)$, then $w^*$ is an infeasibility certificate for $(D)$.

{$(iii)$} Cases (2) and (3) in \Cref{prop:LP_result} establish that the LPs have no optimal solutions by certifying that either $(P)$ or $(D)$ are infeasible. 

\end{remark}

Our second important concept is zero-sum games.
A \emph{zero-sum game} is specified by a matrix $M = (m_{ij})\in \mathbb R^{m\times n}$.
There are two players, a row player and a column player; for ease of description, we will assume that the first is male and the second is female. 
The row player's \emph{action} is to select a row $i$ with $1\le i\le m$ and the column player's action is to select a column~$j$ with $1\le j\le n$.
Selecting row $i$ and column $j$ leads to a payment of the associated matrix entry, $m_{ij}$, from the column player to the row player.
Hence, whereas the row player's goal is to maximize the value of the matrix entry chosen, the column player's goal is to minimize it.  

Players are allowed to randomize their actions.
A \emph{mixed strategy} is a probability distribution over actions.
The row and column players select mixed strategies over rows and columns, respectively.
These can be specified by vectors $p\in \mathbb R^m$ and $q\in \mathbb R^n$ that are nonnegative and satisfy $p\supt \allones = q\supt \allones = 1$.
The entries of $p$ and $q$ are the probabilities with which the respective rows and columns are played.
If the row player plays $p$ and the column player plays $q$, then the row player receives the \emph{expected payment} of $p\supt Mq$ from the column player.

As a solution concept for zero-sum games, we consider maximin strategies.
These are strategies that give a best possible guarantee to a player, even if the opponent plays the worst possible counter-strategy. 
More precisely, a \emph{maximin strategy} for the row player (who tries to maximize his payoff) is a strategy $p^*$ that satisfies
$$p^* \in \arg\max_p\min_q p\supt M q\text.$$

Similarly, a \emph{minimax strategy} for the column player is a strategy $q^*$ that satisfies
$$q^*\in \arg\min_q\max_p p\supt M q$$

If we replace the matrix $M$ by its negative, we see that this is the same as a maximin strategy in the game where the column player tries to maximize the negative payoff she receives.
Therefore, we simplify terminology and say that we compute maximin strategies for both players when we actually mean computing a maximin strategy for the row player and a minimax strategy for the column player.
We call $(p^*,q^*)$ a \emph{pair of maximin strategies}.

The unique payments determined by these strategies, i.e., $u^* = \max_p\min_q p\supt M q$ and $v^* = \min_q\max_p p\supt M q$ are called the \emph{maximin value} and \emph{minimax value} of the game, respectively.
The famous Minimax Theorem of \citet{vNeu28a} states that
$$u^* = \max_p\min_q p\supt M q= \min_q\max_p p\supt M q=v^*\text,$$
 i.e., the maximin value of the game coincides with its minimax value.
Therefore, this quantity is 
called the \emph{value of the game}.
The goal of \emph{solving a zero-sum game} is to compute a pair of maximin strategies.
Note that such strategies are mutual best responses, and therefore constitute a Nash equilibrium.

It is well known that maximin strategies as well as the value of the game can be computed by the following LPs; this is the easy direction of the correspondence between LPs and zero-sum games \citep[see, e.g.,][Equations (7) and (8)]{vSte23a}:
\begin{align*}
                    (&\Row)                &~               (&\Col)       \\
    \text{max} \quad &u  &\text{min} \quad &v           \\
    \text{s.t. }     &M\supt p \geq u \allones   &\text{s.t. }     &Mq \leq v \allones \\
    &\allones\supt p = 1, \; p \geq \allzeros      &   &\allones\supt q = 1, \; q\geq \allzeros
\end{align*}
Note that $(\Row, \Col)$ are a primal-dual LP pair.
The optimal value of the linear program $(\Row)$ is the maximin value corresponding to the row player.
An optimal solution $(u^*,p^*)$ of $(\Row)$ consists of the maximin value together with a maximin strategy.
In fact, the constraints of $(\Row)$ imply that $p$ is a probability distribution and that $p\supt M\ge \allones\supt u$.
Therefore, ${p^*}\supt M q\ge u^*$ for all mixed strategies $q$ of the column player.
Similarly, $(\Col)$ determines the minimax value and a minimax strategy of the column player.

\section{Extending Von Neumann's Reduction}
\label{extending_vN_reduction}

In this section, we provide our extensions of \citeauthor{vNeu53a}'s reduction.  
First, we consider the case where $b > \allzeros$ and $c > \allzeros$, i.e., both vectors are strictly positive.  
Subsequently, we will show how this enables us to tackle the case that the  constraint matrix is nonnegative, i.e., $A  \ge \allzeros$.

\subsection{Positive Objective and Constraint Vectors}\label{sec:posvectors}
The key technical step of our reduction is a rescaling of the LPs $(P)$ and $(D)$; this is where we require 
$b > \allzeros$ and $c > \allzeros$. 
Consider the diagonal matrices $B \in \mathbb R^{m\times m}$ and $C \in \mathbb R^{n\times n}$ defined by

$$ B = 
\begin{pmatrix}
\medskip
\frac 1{b_1} &  &  \\
 & \rotatebox{135}{\dots} &  \\
 & &\frac 1{b_m} \\
\end{pmatrix}
\quad
\text{ and }
\quad
C = 
\begin{pmatrix}
\medskip
\frac 1{c_1} &  &  \\
 & \rotatebox{135}{\dots} &  \\
 & &\frac 1{c_n} \\
\end{pmatrix}\text.$$

We define the matrix $M = BAC \in \mathbb R^{m\times n}$ (its name already suggests that this will be the payoff matrix of our reduced zero-sum game) and a related pair of LPs as follows.
\begin{align*}
    (P')&
    & (D') & \\
    \text{max} \quad \allones\supt \xi&
    &\text{min} \quad \allones\supt \eta& \\
    \text{s.t. } \quad M \xi &\le \allones
    &\text{s.t. } \quad M\supt \eta&\ge \allones  \\
    \xi&\ge \allzeros
    & \eta &\geq \allzeros\\
\end{align*}

It is easy to see that this change of variables by scaling merely leads to a rescaling in the solutions of the original pair of LPs:
    \begin{enumerate}[label=(\roman*)]
        \item $\xi$ (resp., $\eta$) is feasible for $(P')$ (resp., $(D')$) if and only if $C\xi$ (resp., $B\eta$) is feasible for $(P)$ (resp., $(D)$)
    \item
    $(\xi^*, \eta^*)$ is an optimal primal-dual pair for $(P', D')$ if and only if $(C\xi^*,B\eta^*)$ is an optimal primal-dual pair for $(P,D)$.
    \item
    $\bar\xi$ is an unboundedness certificate  for $(P')$ if and only if $C\bar\xi$ is an unboundedness certificate  for $(P)$.
    \end{enumerate}
We remark that $(P)$ and $(P')$ are never infeasible when $b > \allzeros$, since in this case, $x = \allzeros$ is a feasible solution.

Now, let $(G)$ be the zero-sum game with payoff matrix $M$ where the row-player wants to maximize his payoff and the column player wants to minimize her payoff. 
Recall that the maximin and minimax values of this game correspond to the optimal values of the LPs $(\Row)$ and $(\Col)$ defined in \Cref{sec:prelims}.
We have the following theorem.
Note that while it establishes a correspondence between zero-sum games and LP pairs of the type $(P',D')$, it applies to LP pairs of the type $(P,D)$ whenever $b > \allzeros$ and $c > \allzeros$, i.e., the rescaling can be performed.

\begin{theorem}\label{thm:main}
    \begin{enumerate}
    \item
    Let $(p^*,q^*)$ be a pair of maximin strategies of the zero-sum game $(G)$ and $v^* (=u^*)$ its value.
       \begin{enumerate}
        \item 
        If $v^*>0$, then 
        $\left(\frac 1{v^*}q^*,\frac 1{v^*}p^*\right)$ is an optimal primal-dual pair for $(P',D')$.
        \item 
         If $v^* \leq 0$, then 
        $q^*$ is an unboundedness certificate for  $(P')$.
         \end{enumerate}
         \item Let $v^*$ be the value of the zero-sum game $(G)$.
         \begin{enumerate}
             \item 
         If $(\xi^*,\eta^*)$ is an optimal primal-dual pair for $(P')$ and $(D')$, then ${\allones\supt \xi^*}\left(={\allones\supt \eta^*}\right) > 0$,
         $\left(\frac{1}{\allones\supt \eta^*}\eta^*, \frac{1}{\allones\supt \xi^*}\xi^*\right)$ 
         is a pair of maximin strategies of the zero-sum game $(G)$, and $v^*=\frac{1}{\allones\supt \xi^*}$. 
            \item 
            If $(P')$ is unbounded, then $v^* \le 0$. 
         \end{enumerate}
    \end{enumerate}
\end{theorem}

\begin{proof}
We prove the statements one by one.
    \begin{enumerate}
  
   \item[{1-(a).}]
       The feasibility of $\left(\frac 1{v^*}q^*,\frac 1{v^*}p^*\right)$
       for $(P')$ and $(D')$ is readily verified by the feasibility of 
       $(p^*,q^*)$ for $(\Col)$ and $(\Row)$.
       In addition, 
       $\allones\supt p^*= \allones\supt q^*=1$ implies 
       $\allones\supt \frac{1}{v^*}q^*= \allones\supt \frac{1}{v^*}p^*$ which completes the proof.
     \item[{1-(b).}] Suppose that $v^*\le 0$. 
         Then, $Mq^*\leq v^* \allones \leq \allzeros$, $\allones\supt q^* =1>0$, and $q^* \geq \mathbf{0} $. 
     Hence, $q^*$ is an unboundedness certificate for $(P')$.
    \item[{2-(a).}] 
    First, note that the optimal objective value of $(P')$ is always strictly larger than $0$, and, therefore, $\allones\supt \xi^* > 0$.
	Define $\hat v := \frac{1}{\allones\supt \xi^*} > 0$.
    We have that 
    $\allones\supt (\hat v\eta^*)=\allones\supt (\hat v\xi^*)=1$, $\hat v\eta^*\ge \allzeros$, and $\hat v\xi^*\ge \allzeros$, and, therefore, $\hat v\eta^*$ and $\hat v\xi^*$ are valid strategies.
    
	Next, we show that $\hat v\eta^*$ and $\hat v\xi^*$ are mutual best responses.
    Consider any strategy $q$ of the column player.
    By dual feasibility, it holds that 
	\begin{equation}\label{eq:LBvalue}
		(\hat v\eta^*)\supt M q \ge (\hat v\allones)\supt q = \hat v\text.
	\end{equation}
	
    There, $(\hat v\allones)\supt q = \hat v$ holds because $q$ is a probability distribution.
    Moreover, by primal feasibility, 
	\begin{equation}\label{eq:UBvalue}
		(\hat v\eta^*)\supt M (\hat v \xi^*) \le (\hat v\eta^*)\supt  \hat v \allones = \hat v \hat v {\eta^*}\supt  \allones = \hat v \hat v\frac 1{\hat v} = \hat v\text.
	\end{equation}
    Hence, $\hat v\eta^*$ is a best response for the row player.
    The computation for the column player is analogous.
    Together, $(\hat v\eta^*,\; \hat v\xi^*)$ is a Nash equilibrium and therefore a pair of maximin strategies for $(G)$.
	
	Finally, combining \Cref{eq:LBvalue} for $q = \hat v \xi^*$ and \Cref{eq:UBvalue}, we know that $(\hat v\eta^*)\supt M (\hat v \xi^*) = 
\hat v$.
	Hence, $v^* = \hat v$ is the value of the game $(G)$.
    \item[{2-(b).}]
    Assume that $w^*$ is an unboundedness certificate for $(P')$, i.e., $Mw^*\le \allzeros$, $w^*\ge\allzeros$, and $\allones\supt w^* > 0$.
    Let $\alpha = \frac{1}{\allones\supt w^*} >0$ and consider $q = \alpha w^*$ and $v = 0$.
    We claim that $(q,v)$ is feasible for $(\Col)$.
    
    Indeed, since $w^*\ge \allzeros$, it holds that $\alpha w^*\ge \allzeros$.
    In addition, $\allones\supt q = \alpha \allones\supt w^* = 1$.
    Moreover, $Mq = \alpha M w^* \le \alpha \allzeros = v \allones$.

    Hence, $(q,v)$ is feasible for $(\Col)$ and therefore $v^*\le 0$.
    \qedhere
    \end{enumerate}
\end{proof}

We conclude this subsection with remarks about the interpretation of our reduction. 

\paragraph{Linear Programs in Game Form}
An appealing aspect of our reduction is that it uses a scaled version of the constraint matrix, $A$, as the game matrix. The scaling is done as follows: the columns of $A$ are scaled by the reciprocals of the coefficients in the objective function, $c$, and the rows of $A$ are scaled by the reciprocals of the coefficients in the right hand side, $b$.
In addition to being natural, this approach strengthens the connection between the original linear program and the reduced game beyond what was achieved by Dantzig and his successors.

Once the rescaling is applied, we obtain a correspondence by directly using the constraint matrix as the game matrix of a zero-sum game.
In other words, LPs of the form $(P',D')$ can be interpreted as \emph{linear programs in game form}.
Hence, von Neumann's hide-and-seek game results from the realization that the assignment problem can be written in game form.
And more generally, we observe that any linear program pair of the form $(P,D)$ with $b > \allzeros$ and $c > \allzeros$ can be transformed to game form, as well.

Unlike in previous reductions, the column and row players in our reduction have different roles corresponding to the primal and dual linear programs, respectively.
Each fixed strategy (including a nonoptimal one) gives a guarantee to a player, i.e., a minimum reward for the row player against adversarial play of the column player or a maximum payment by the column player against an adversarial row player. 
These guarantees are obtained by solving $(\Row)$ and $(\Col)$ after fixing strategies $p$ and $q$, respectively.
When players play optimal strategies, they get their optimal guarantee, which is exactly the maximin or minimax value.

In fact, one can view playing the game of an LP in game form as \emph{playing the linear program}.
Any strategy guaranteeing the column player nonpositive payoff (not only the maximin strategy) gives rise to an unboundedness certificate as in Part~{\it 1-(b)} of the proof of \Cref{thm:main}.
Moreover, strategies guaranteeing the row player a positive payoff correspond to dual feasible points.
Once again, this works even for nonoptimal strategies, but can then lead to feasible points with a suboptimal objective value.
Hence, players compete by proposing (good) feasible points of their respective LPs, whenever this is possible. 
This can be viewed as playing an LP.
The reduced zero-sum game yields a reasonable outcome of playing, even for suboptimal or infeasible strategies.
In the special case, where both players play maximin strategies, we obtain an optimal primal-dual pair.

\paragraph{Value of Game as Degree of Infeasibility}

The games by \citeauthor{Dant51a} and \citeauthor{Adle13a} are symmetric and have value zero, i.e., the value of the game has no information.\footnote{While we regard the nonsymmetry of our reduced games as beneficial for interpreting the game, we acknowledge that reductions producing symmetric games have other, computational advantages.
For instance, the symmetrization techniques by \citet{GKT50a} imply that algorithms for solving symmetric games can be applied to nonsymmetric games without significant further computational effort.} 
Moreover, the value of the reduced games by \citeauthor{vSte23a} as well as \citeauthor{BrRe23a} are zero if and only if there exists an optimal solution to the original LP. 
In contrast, the value of our reduced game has valuable information: first of all, it allows us to decide whether the LP admits an optimal solution and, if so, it is the reciprocal of its optimal value.

    Additionally, the primal-dual pair $(\Row,\Col)$ always has primal and dual feasible solutions, whereas $(D')$ might be infeasible.
    Recall from \Cref{rem:LPremarks} that an infeasibility certificate for $(D')$ is equivalent to an unboundedness certificate for $(P')$, i.e., we ask whether there exists $w^*$ with $Mw^*\le \allzeros$, $w^*\ge \allzeros$, and $\allones\supt w^* > 0$.
    By scaling $w^*$, the last condition can equivalently be written as $\allones\supt w^* = 1$.
    The LP $(\Col)$ only differs from these constraints by adding the minimum necessary slack to the constraints governed by $v$.
    Hence, the value of the reduced zero-sum game can be interpreted as a \emph{degree of feasibility}, i.e., the minimum slack needed to guarantee feasibility of $(D')$.
    Games with a nonpositive value indicate infeasibility of $(D')$, and their exact value can then be seen as measuring a \emph{degree of infeasibility}.

\subsection{Nonnegative Constraint Matrix}
Finally, we show that the developed theory in \Cref{sec:posvectors} also allows us to deal with linear program pairs $(P,D)$ when $A \ge \allzeros$ but $b$ and $c$ are not constrained to be positive.  

The intuitive idea is to preprocess the LP to satisfy the conditions that $b$ and~$c$ are positive:
If $b_i < 0$, then the nonnegativity of $A$ immediately implies infeasibility.
Otherwise, we want to delete all columns and rows that correspond to nonpositive entries of $b$ and $c$.
If $b_i = 0$, then $x_j = 0$ whenever there exists $j\in [n]$ with $a_{ij} > 0$.
Hence, we can delete the $i$th row together with all such columns, and later on insert $x_j = 0$ when reconstructing solutions.
Similarly, $c_j\le 0$ requires $x_j = 0$ and we can omit such columns to transition to a smaller LP.
We thus reach a reduced LP with $b > \allzeros$ and $c > \allzeros$, for which we can apply the correspondence of \Cref{sec:posvectors}.

We now define sets $I_0$, $J_0$, and $J_{-}$ of indices of rows and columns to be deleted according to the above discussion.
\begin{align}
    I_0 &:= \{ i\in [m] \; |\; b_i =0 \},\notag\\
    J_0 &:= \{ j\in [n] \; |\; a_{ij}>0 \text{ for some } i \in I_0\}, \textnormal{ and }\label{eq:LPremove}\\
    J_{-} &:= \{ j\in [n] \; |\; c_j \leq 0 \}\notag
\end{align}

Let $\hat{A}$, $\hat{b}$, and $\hat{c}$ be obtained by removing rows indexed by $I_0$ and columns indexed by $J_0\cup J_{-}$ from $A$, indices in $I_0$ from $b$, and indices in $J_0\cup J_{-}$ from $c$, respectively.
This gives rise to the modified LP pair $(\hat P, \hat D)$ defined as follows.

\begin{align*}
    (\hat P)&
    & (\hat D) & \\
    \text{max} \quad \hat c\supt x &
    &\text{min} \quad \hat b\supt y & \\
    \text{s.t. } \quad \hat A x &\le \hat b
    &\text{s.t. } \quad \hat A\supt y&\ge \hat c \\
    x&\ge \allzeros
    & y &\geq \allzeros\\
\end{align*}

We solve the original LP, for which we assumed that $A \ge \allzeros$, by applying \Cref{alg:Anonneg} following the strategy discussed above.
First, we check whether $b\ge \allzeros$ and if not, we find an infeasibility certificate for $(P)$.
Otherwise, we 
consider the restricted LP pair $(\hat P, \hat D)$.
However, we have to be careful as the smaller LP might conceal some easy cases that can be detected before deleting rows and columns.
First, we might find an unboundedness certificate (\cref{ln:trivUB}).
Moreover, it can happen that $(\hat P, \hat D)$ is an empty LP as we might remove all rows or columns.
In this case, we find an optimal primal-dual pair (\cref{ln:trivOPT}).

These three steps can be seen as \emph{preprocessing} before we apply our actual correspondence.
Now, we can consider the restricted LP pair $(\hat P, \hat D)$ in \cref{ln:modLP} of the algorithm.
For this LP, we have achieved that $\hat b > \allzeros$ and $\hat c > \allzeros$.
Indeed, we know after preprocessing that $b\ge \allzeros$.
Moreover, removing indices in $I_0$ guarantees that $b > \allzeros$ and removing indices in $J_{-}$ guarantees that $c > \allzeros$.
Hence, we can solve $(\hat P, \hat D)$ by first rescaling the rows and columns to obtain a constraint matrix in game form, and then applying our correspondence with a zero-sum game using \Cref{thm:main}.
As a last step, we insert appropriate values for the omitted indices to obtain certificates for $(P,D)$.
This part can be seen as a \emph{postprocessing} of the solution.
Clearly, \Cref{alg:Anonneg} runs in strongly polynomial time.
We collect the key insights for its correctness in the following theorem.

\begin{algorithm}[tb!]
  \caption{Solving linear programs with $A \ge \allzeros$\label{alg:Anonneg}}
	\KwIn{Linear program pair specified by $A \ge \allzeros$, $b\in \mathbb R^m$, and $c\in \mathbb R^n$}
	\KwOut{Certificate according to \Cref{prop:LP_result}}

\If {there exists $i \in [m]$ such that $b_i < 0$}
{Define $z^*$ by
  $z^*_k
=
	\begin{cases}
		1  & k=i \\
		0  & k \neq i 
	\end{cases}
$.  \\

\KwRet {Infeasibility certificate $z^*$}\label{ln:INF}}
Define $I_0$, $J_0$, and $J_{-}$ as in \Cref{eq:LPremove}. \\
\If {there exists $j \in [n]\setminus J_{-}$ such that $a_{ij} = 0$ for all $i\in [m]$\label{ln:UBcond}}
{Define $w^*$ by
  $w^*_k
=
	\begin{cases}
		1  & k=j \\
		0  & k \neq j 
	\end{cases}
$.  \\

\KwRet {Unboundedness certificate $w^*$}\label{ln:trivUB}}
Define $\numbpara := \max\left\{\left.\frac{c_j}{a_{ij}}\; \right|\; j\in J_0, i\in I_0 \text{ with } a_{ij} > 0 \right\}$ if $J_0 \neq \emptyset$ and $\numbpara = 0$ if $J_0 = \emptyset$.\\
\If{$I_0 = [m]$ or $J_0\cup J_{-} = [n]$}
{
Define $(x^*,y^*)$ by $x^* = \allzeros$ and $y^*_i := \begin{cases}
    \numbpara & i\in I_0\\
    0 & i\notin I_0
\end{cases}$.\\
\KwRet{Optimal primal-dual pair $(x^*,y^*)$}\label{ln:trivOPT}}
Solve linear program pair $(\hat P,\hat D)$ based on $\hat{A}$, $\hat{b}$, and $\hat{c}$.\label{ln:modLP}\\
\If{$(\hat x,\hat y)$ is optimal primal-dual pair}
{
Define $(x^*,y^*)$ by $x^*_j := \begin{cases}
    0 & j\in J_0\cup J_{-}\\
    \hat x_j & j\notin J_0\cup J_{-}
\end{cases}$ and $y^*_i := \begin{cases}
    \numbpara & i\in I_0\\
    \hat y_i & i\notin I_0
\end{cases}$.\\
\KwRet{Optimal primal-dual pair $(x^*,y^*)$\label{ln:OPT}}
}
\end{algorithm}

\clearpage

\begin{theorem}\label{thm:A_nonnegative}
Assume that $A \geq \allzeros$.
    \begin{enumerate}
    \item
    $(P)$ is feasible if and only if $b \geq \allzeros$.
    If there exists $i \in [m]$ such that $b_i < 0$, then \Cref{alg:Anonneg} returns an infeasibility certificate for $(P)$ in \cref{ln:INF}.
    \end{enumerate}
    Now suppose $b \geq \allzeros$.
    \begin{enumerate}
    \setcounter{enumi}{1}
\item $(P)$ is unbounded if and only if there exists $j \in [n]\setminus J_{-}$ such that $a_{ij} = 0$ for all $i\in [m]$. In this case, \Cref{alg:Anonneg} returns an unboundedness certificate for $(P)$  in \cref{ln:trivUB}.
\item Otherwise $(P,D)$ admits an optimal primal-dual solution. There are two cases for \Cref{alg:Anonneg}:
    \begin{enumerate}
        \item If $I_0 = [m]$ or $J_0\cup J_{-} = [n]$, 
        then \Cref{alg:Anonneg} returns an optimal primal-dual pair for $(P,D)$  in \cref{ln:trivOPT}.
        \item  Otherwise, $(\hat P, \hat D)$ admits an optimal primal-dual pair.
        $(\hat P, \hat D)$, 
        and \Cref{alg:Anonneg} returns an optimal primal-dual pair for $(P,D)$  in \cref{ln:OPT}.
    \end{enumerate}
\end{enumerate}
 \end{theorem}
 \begin{proof}
We prove the statements one by one.
     \begin{enumerate}
   \item
   If $b\ge \allzeros$, then $x = \allzeros$ is a feasible solution of $(P)$.
   Conversely, if $x$ is a feasible solution of $(P)$, then $b\ge Ax\ge \allzeros$, where the second inequality is true because $A\ge \allzeros$ and $x\ge \allzeros$. 
    
    Moreover, assume that there exists $i \in [m]$ such that $b_i < 0$ and consider $z^*$ as defined in \Cref{alg:Anonneg}.
    Then $A\supt z^* \geq \allzeros$, $z^*\geq \allzeros$, and $b\supt z^*<0$.
    Hence, $z^*$ is an infeasibility certificate for~$(P)$.
    \item Assume that there exists $j \in [n]\setminus J_{-}$ such that $a_{ij} = 0$ for all $i\in [m]$. 
    Consider $w^*$ as returned in \cref{ln:trivUB}.
    Clearly $w^* \ge \allzeros$. In addition, since $j\notin J_{-}$, it holds that $c_j > 0$ and therefore $c\supt w^* > 0$.
    Finally, for every $i\in [m]$ it holds that every $(Aw^*)_i = a_{ij} = 0$. 
    We used that only the $j$th entry of $w^*$ is nonzero and our assumption on the entries of $A$.

    Conversely, assume that $(P)$ admits an unboundedness certificate $w^*$, i.e., $w^*\ge \allzeros$, $Aw^* \le \allzeros$, and $c\supt w^* > 0$.
    Since $c\supt w^* > 0$ and $w^*\ge\allzeros$, there exists $j^*\in [n]$ with $c_{j^*} > 0$ and $w^*_{j^*}>0$.
    By definition $j^*\notin J_{-}$.
    Define $A_{j^*}$ as the $j^*$th column of $A$.
    Then, $\allzeros \ge Aw^* \ge w^*_{j^*}A_{j^*}$. 
    Since $w^*_{j^*}>0$, it follows that $A_{j^*} \le \allzeros$ and therefore $A_{j^*} = \allzeros$ by the nonnegativity of~$A$.
    \item[3-(a).] Consider the pair 
        $(x^*,y^*)$ returned in \cref{ln:trivOPT}.
        Clearly $x^* = \allzeros$ is primal feasible.
        Moreover, $y^* \ge \allzeros$.
        Now let $j\in [n]$.
        If $j\in J_{-}$, the $j$th constraint is satisfied.
        Otherwise, it holds that $j\in [n]\setminus J_{-}$ and we have already excluded the case where $a_{ij} = 0$ for all $i\in [m]$.
        Hence, there exists $i\in [m]$ with $a_{ij} > 0$.
        By definition of $\numbpara$, it holds that $\numbpara \ge \frac{c_j}{a_{ij}}$.
        Therefore,
        $(A\supt y^*)_j \ge a_{ij} y^*_i = a_{ij} \numbpara \ge  a_{ij}\frac{c_j}{a_{ij}} = c_j$.
        Hence, $y^*$ is dual feasible.
        Finally, it holds that $b\supt y^* = \allzeros = c\supt x^*$ where the first inequality holds by definition of $y^*$ (whose components are only nonzero for indices in $I_0$).
        Together, $(x^*,y^*)$ is an optimal primal-dual pair.

    \item[3-(b).]
    We will first prove the existence of feasible solutions to $\hat P$ and $\hat D$.
    Clearly $\hat x = \allzeros$ is feasible for $(\hat P)$ because $\hat b > \allzeros$.
    
    Moreover, since $\hat b\ge \allzeros$, we can apply part \textit{2-(a)} to conclude that $(\hat P)$ is bounded if $\hat A$ does not contain $0$-columns (which is a stronger condition than only requiring certain nonzero columns).
    Let $j\in [n]\setminus (J_0\cup J_{-})$.
    We claim that there exists $i\in [m]\setminus I_0$ with $a_{ij} > 0$.
    Assume that this was not the case, i.e., $a_{ij} = 0$ for all $i\in [m]\setminus I_0$.
    Note that since $j\notin J_0$, we also have that $a_{ij} = 0$ for all $i\in I_0$.
    However, this would imply that we have identified a $j\in [n]\setminus J_{-}$ such that $a_{ij} = 0$ for all $i\in [m]$, contradicting that \Cref{alg:Anonneg} has reached \cref{ln:modLP}.
    Hence, $(\hat P)$ is feasible and bounded, and, therefore, $(\hat D)$ is feasible as well.
    By strong LP duality, we have the existence of an optimal primal-dual pair.

    Now, assume that $(\hat x, \hat y)$ is an optimal primal-dual pair for 
        $(\hat P, \hat D)$ and consider
        $(x^*,y^*)$ as defined in \Cref{alg:Anonneg}.
Then, $x^*\ge \allzeros$ and $y^* \ge \allzeros$.
In addition, whenever $b_i = 0$, then all column indices $j\in [n]$ with $a_{ij} > 0$ are in $J_0$ therefore $(Ax^*)_i = 0$ and the $i$th constraint is satisfied.
All other constraints of $(P)$ are satisfied because their corresponding constraint in the restricted game is satisfied.
Hence, $x^*$ is feasible for $(P)$.

Moreover, for all indices $j\in J_{-}$, it holds that $(A\supt y^*)_j \ge 0 \ge c_j$ where we use the nonnegativity of $A$ and $y^*$ for the first inequality.
Additionally, if $j\in J_0$, then there exists $i\in I_0$ with $a_{ij} > 0$.
By definition of $\numbpara$, it holds that $\numbpara \ge \frac{c_j}{a_{ij}}$.
Hence,
$(A\supt y^*)_j \ge a_{ij} y^*_i = a_{ij} \numbpara \ge  a_{ij}\frac{c_j}{a_{ij}} = c_j$.
It follows that $y^*$ is dual feasible.

Finally, it holds that 
\begin{equation}\label{eq:optimal}
    c\supt x^* = \hat c\supt \hat x = \hat b\supt \hat y = b\supt y^*\text.
\end{equation}
The first equality holds because we have extended $\hat x$ only by inserting zeros.
The second equality holds by primal-dual optimality of $(\hat x, \hat y)$.
The third inequality holds because $b_i = 0$ whenever $i\in I_0$.
Hence, \Cref{eq:optimal} implies that $(x^*,y^*)$ is an optimal primal-dual pair.
    \qedhere
    \end{enumerate}

 \end{proof}

We remark that solving $(\hat P, \hat D)$ could a priori lead to the existence of either an optimal primal-dual pair or an unboundedness certificate, while $(\hat P)$ is never infeasible (see the discussion of the rescaled game in \Cref{sec:posvectors}).
However, unboundedness of $(\hat P)$ never occurs, as we prove in the third part of \Cref{thm:A_nonnegative}.
Since \Cref{thm:A_nonnegative} covers all cases where \Cref{alg:Anonneg} produces an output, it proves the correctness of \Cref{alg:Anonneg}.

\begin{corollary}
    \Cref{alg:Anonneg} solves linear programs with $A \ge \allzeros$.
\end{corollary}

\section{Discussion}

The question of generalizing von Neumann's result \citep{vNeu53a} to an arbitrary primal-dual pair of LPs is clearly an important one. 
Our work explores an extension based on rescaling, the crucial idea that helps define the payoffs in the hide-and-seek game. 
While at first glance this seems only possible if $b > \allzeros$ and $c > \allzeros$, we managed to reduce the case $A \geq \allzeros$ to this case as well. 
Going further may well require a new approach. 

We conclude our paper with discussions tying up a few loose ends.

\subsection{Production Planning Games}\label{sec:productionapplication}

We want to demonstrate the applicability of our reductions by discussing an interpretation of the reduced game for any LP with positive $b$ and~$c$ and demonstrate how \citeauthor{vNeu53a}'s hide-and-seek game is a special case.

LPs with positive $b$ and $c$ 
 can be viewed as planning optimal production:
A central authority decides how much of each of $n$ goods to produce, where $x_j$ is the amount of good~$j$ produced.
The profit from producing one unit of good~$j$ is $c_j$.
The production of each good consumes certain amounts of $m$ resources representing, e.g., raw materials and electricity, or even labor hours and machine capacity.
Producing a unit of the $j$th good requires $a_{ij}$ units of the $i$th resource.
Finally, we have total availability of $b_i$ units for the $i$th resource.
The total usage of resources is not allowed to exceed their availability, which is exactly captured by the constraints $Ax\le b$.
Furthermore, it is reasonable to assume positive profits from producing a good as well as a positive availability of each resource.
Hence, the production planning scenario is captured by the linear program $(P)$ for positive $b$ and $c$.
In many applications, we also have nonnegative entries in $A$.
However, negative entries in $A$ allow for modeling further scenarios, such as producing a resource, e.g., electricity, as a byproduct in a production process.

The reduced game now has the following interpretation: the column player mimics the production planner while the row player mimics a resource provider.
A pure strategy for the column player is to pick a good while the row player picks a resource.
While the production planner wants to gain high production profits, the resource player wants to sell them a scarce resource.
Selecting the $j$th good and $i$th resource leads to a payment by the column player to the row player of $\frac {a_{ij}}{b_ic_j}$.
This quantity is the resource share needed to produce a unit of the $j$th good, that is, $\frac {a_{ij}}{b_i}$, divided by the profit $c_j$ for producing one unit.
Together, this is the required resource share per unit profit.
If the selected good leads to a small profit or needs a high amount of the resource, i.e., the resource is scarce for this good, then the payment is high.
Conversely, needing a small share of the resource or making a lot of profit leads to a small payment, or even a negative payment if the resource is produced as a byproduct indicated by a negative $a_{ij}$.
Hence, it makes intuitive sense that the payoff depends on the reciprocals of both the resource availability and the price of the good.

This perspective also offers a novel view on the hide-and-seek game. 
In this context, the production planner generates edges. 
For instance, the planner could be an event manager organizing tennis matches. 
A match between players $i$ and $j$ yields a surplus $\surplus_{ij}$. 
The resource provider supplies the players, who correspond to the vertices in the assignment problem. 
By the rules of the hide-and-seek game, if a player participates in a desired match, then a compensation of $1/\surplus_{ij}$ must be paid. 
Assuming that the surplus from a match accrues linearly over its duration, with a total surplus of $\surplus_{ij}$, this compensation corresponds to the fraction of the match that yields one unit of profit.

\subsection{Computational Considerations}

The key idea of the papers repairing \citeauthor{Dant51a}'s original reduction is to modify his reduced game: \citet{Adle13a} and \citet{vSte23a} add an additional row while \citet{BrRe23a} transform the game with elementary column operations.
These reductions have in common the fact that the constraint matrix or vectors in the reduced games need large numbers, dependent on the input numbers and exponential in $m$ or~$n$.\footnote{\label{fn:largenumb}
As an example, in the reduction by \citet{Adle13a}, the additional entry of the constraint vector is estimated as a number that is at least $m(\sqrt{m}\theta)^m+1$ where $\theta$ is the absolute value of the largest number contained in $A$ and $b$.}
Nonetheless, since the large numbers can be computed with a polynomial number of arithmetic operations and have a bit size polynomial in the input bit size, the reductions are strongly polynomial.\footnote{For a formal definition of strongly polynomial time, we refer to \citet{GLS93a}.}
In addition, \citet{Adle13a} also provides a Turing reduction based on sequential variable elimination that avoids large numbers.
Our reduction is a strongly polynomial reduction without large numbers.

Within the different realm of computational models based on Turing machines, all of the aforementioned reductions are log-space reductions, i.e., they only use logarithmic space except for outputting the polynomially large reduced instance \citep{ArBa09a}. 
In fact, the occurring large numbers can be computed bit-wise in logarithmic space. 
This has the interesting consequence that computing maximin strategies in zero-sum games is P-complete because solving linear programs is P-complete \citep{DLR79a}.
Apparently, this observation was missed in previous work.

\subsection{Relationship with Dantzig's Reduction}

The reduction in the seminal paper by \citet{Dant51a} is based on the following skew-symmetric game matrix, given the input data of the primal LP $(P)$.\footnote{\citet{Dant51a} attributes the idea for this matrix to Tucker.}

\begin{equation}\label{eq:Dantzig}
\begin{pmatrix}
\medskip
\allzeros & A & -b \\
-A\supt & \allzeros & c \\
b\supt & -c\supt & 0\\
\end{pmatrix}\text.
\end{equation}

Denoting a maximin strategy as $(p^*, q^*, t^*)$, the key result of the paper by \citet{Dant51a} states:
\begin{enumerate}
\item
If $ t^*>0$, then $(x^*,y^*)= \left(\frac{1}{t^*}q^*,\frac{1}{t^*}p^*\right) \text{ satisfies (1) in \Cref{prop:LP_result}}.$
\item
    If $b\supt p^* -c\supt q^* <0$, then either $z^*=p^*$ satisfies (2) or $w^*=q^*$ satisfies (3) in \Cref{prop:LP_result}.
\end{enumerate}

Note that since $(p^*,q^*,t^*)$ is a maximin strategy for a game of value~$0$, it holds that $t^*(b\supt p^* -c\supt q^*) = 0$. 
The goal of solving $(P)$ is achieved, except for the case $t^*=b\supt p^* -c\supt q^* =0$. 
Notably, this ``hole'' in the reduction was clearly recognized by \citet{Dant51a}.

In this section, we explore whether the restrictions that we impose on linear programs suffice to make his reduction work.
First, we show that under the assumption  $b>\allzeros$ and $c>\allzeros$, his reduction works.
Let $(p^*,q^*,t^*)$ be a minimax strategy of the row player.
Then, it holds that

{\small
  \begin{equation*}
  \begin{array}{lll}
     (i)\   A q^* -b t^*  \leq \allzeros, 
       &(ii)\  -A\supt p^*  +c t^* \leq \allzeros, 
       &(iii)\   b\supt p^*  - c\supt q^* \leq \allzeros, \\ 
     (iv)\    \allones\supt p^* +  \allones\supt q^*  + t^* =1, 
     &(v)\   p^* \geq \allzeros, \; q^*\geq \allzeros,\; t^* \geq 0, 
     &(vi)\ t^*(  b\supt p^*  - c\supt q^* )= 0.
  \end{array}
  \end{equation*}}

 Observe that, as the game is symmetric, $(p^*,q^*,t^*)$ is also a maximin strategy for the column player, and the value of the game is $0$.

\begin{proposition}
\label{prop:dantzig}
   Assume that $b>\allzeros$ and $c>\allzeros$ and let $(p^*,q^*,t^*)$ be a minimax strategy of the row player in Dantzig's game given by \Cref{eq:Dantzig}. Then the following statements are true:
   \begin{enumerate}
 \item If $t^*>0$, then $\left(\frac{1}{t^*}q^*, \frac{1}{t^*}p^*\right)$ is an optimal primal-dual pair for $(P,D)$.
 \item If $t^*=0$, then $q^*$ is an unboundedness certificate for  
 $(P)$.

   \end{enumerate}

\end{proposition}

\begin{proof}
We prove the statements one by one.
    \begin{enumerate}
    \item We assume $t^* >0$.
  Conditions ($i$), ($ii$), and ($v$) imply that $\frac{1}{t^*}q^*$ and $\frac{1}{t^*}p^*$ are feasible for $(P)$ and $(D)$, respectively. 
  In addition, by ($vi$), $\frac{1}{t^*}\left(b\supt p^*  - c\supt q^*\right) = 0$, which completes the necessary and sufficient conditions for the optimality of $\frac{1}{t^*}q^*$ and~$\frac{1}{t^*}p^*$.
  \item Now assume $t^* = 0$.
  Assume for contradiction that $q^* = \allzeros$. 
  Then, by ($iv$) and ($v$), it holds that $0 \neq p^* \ge \allzeros$.
  Hence, since $b > \allzeros$, it follows that $b\supt p^* > 0$, contradicting ($iii$). 
  Thus, $q^* \neq 0$ which implies (as $c>\allzeros$) that
  $c\supt q^* >0$.
  Moreover, by ($i$) and ($v$), it holds that $Aq^* \leq \allzeros$ and $q^* \ge \allzeros$.
  Together, $q^*$ is an unboundedness certificate for $(P)$.\qedhere
  \end{enumerate}
\end{proof}

Notably, the second part of \Cref{prop:dantzig} has the precondition $t^* = 0$ whereas the theorem by \citet{Dant51a} had the precondition $b\supt p^* - c\supt q^* < 0$.
Recall that the issue with Dantzig's proof was that it could happen that both $t^* = 0$ and $b\supt p^* - c\supt q^* = 0$.
This raises the question whether \Cref{prop:dantzig} works because this case never occurs or whether it works because this case can be dealt with successfully (unlike in Dantzig's reduction).
Remarkably, the latter is the case, as our next example shows.
Regardless of conditions on $p^*$ and $q^*$, we obtain an unboundedness certificate whenever $t^* = 0$.

\begin{example}
    We provide an example illustrating \Cref{prop:dantzig} for the case when $t^* = 0$ and $b\supt p^* - c\supt q^* = 0$.
    Consider the LP given by
    \begin{equation*}
        A = \begin{pmatrix}
            -1 & \frac 12 \\
            1 & -\frac 12
        \end{pmatrix}\text{, }
        b = \begin{pmatrix}1\\1\end{pmatrix}\text{, and }
        c = \begin{pmatrix}1\\1\end{pmatrix}\text.
    \end{equation*}
    It is easy to verify that $(p^*,q^*,t^*) = \left(\frac 14, \frac 14, \frac 16,\frac 13, 0\right)$ is a maximin strategy of Dantzig's reduced zero-sum game. 
    Clearly, it holds that $t^* = 0$ and $b\supt p^* - c\supt q^* = 0$.
    However, this does not pose a problem as we can still extract the unboundedness certificate $q^* = \left(\frac 16,\frac 13\right)$.\hfill$\lhd$
\end{example}

In contrast to the case $b> \allzeros$ and $c> \allzeros$ covered by \Cref{prop:dantzig}, it is possible for Dantzig's reduction to fail whenever $A \geq \allzeros$.

\begin{example}
    Consider any constraint matrix $A \geq \allzeros$ and consider the linear program pair given by $A$ as well as $b = \allzeros$ and $c = \allzeros$. 
    \citeauthor{Dant51a}'s reduced game as in \Cref{eq:Dantzig} for this case is
    \begin{equation*}
        \begin{pmatrix}
        \allzeros & A & \allzeros \\
        -A\supt & \allzeros & \allzeros \\
        \allzeros\supt & \allzeros\supt & 0\\
        \end{pmatrix}\text.
    \end{equation*}
    Clearly, $(p^*, q^*, t^*) = (\unitvect,\allzeros,0)$ is a maximin strategy for both players, where $\unitvect = (1,0,\dots,0)$ is the first unit vector.
    Since, $b = \allzeros$ and $c = \allzeros$, it holds that $t^* = 0$ and $b\supt p^* - c\supt q^* = 0$, so we cannot directly apply \citeauthor{Dant51a}'s correspondence.
    Furthermore, extracting a primal-dual optimal pair or certificates of infeasibility or unboundedness from the maximin strategy like in \Cref{prop:dantzig} fails.
    Scaling $p^*$ is never dual feasible (apart from scaling with $0$), and $p^*$ is also not a certificate of infeasibility or unboundedness.
    \hfill$\lhd$
\end{example}

Arguably, the case of $b = \allzeros$ and $c = \allzeros$ in the last example may seem unsatisfying.
However, we want to stress that \citeauthor{Dant51a}'s correspondence can only fail for somewhat trivial cases, namely cases covered by our preprocessing in \Cref{alg:Anonneg}.
Indeed, instead of using our reduction, we can also use \citeauthor{Dant51a}'s reduction in \cref{ln:modLP} of \Cref{alg:Anonneg} for tackling LPs with $A\ge \allzeros$.
This approach is sound due to \Cref{prop:dantzig}.

\section*{Acknowledgements}
We would like to thank the editor and anonymous referees from GEB for their excellent comments.
Most of this research was done while Martin Bullinger was affiliated with the University of Oxford.
Martin Bullinger was supported by the AI Programme of The Alan Turing Institute. Vijay V. Vazirani was supported by NSF grant CCF-2230414.


\end{document}